\documentclass{article}
\usepackage{fullpage}
\usepackage[utf8]{inputenc}

\usepackage[english]{babel}
\usepackage{amsmath}
\usepackage{graphicx}
\usepackage{tikz}
\usetikzlibrary{shadows}
\usepackage{hyperref}
\usepackage{pgffor}
\usepackage{pdflscape} 
\usepackage[T1]{fontenc}
\usepackage[outline]{contour}
\usepackage[english]{babel}
\usepackage{amsmath}
\usepackage{amsfonts}
\usepackage{titletoc}
\usepackage{tikz}
\usepackage{mathtools}
\usepackage{fp}
\usepackage{todonotes}
\usepackage{pseudo}
\usepackage{mdframed}

\usepackage{bold-extra}
\usepackage{algorithm}
\usepackage[noend]{algorithmic}
\usepackage{amsthm, amsmath, amssymb, amsfonts} 
\usepackage{siunitx}
\usepackage{bm}
\usepackage{array}
\usepackage{booktabs}
\usepackage{xcolor}
\usepackage{orcidlink}
\usepackage{thmtools} 
\usepackage{thm-restate}
\newtheorem{theorem}{Theorem}[section]
\newtheorem{lemma}[theorem]{Lemma}

\usetikzlibrary{calc}
\usetikzlibrary{decorations.pathreplacing}

\title{On the Hardness of the Drone Delivery Problem}
\author{
    Simon Bartlmae$^{1}$\orcidlink{0009-0009-6953-8347}, 
    Andreas Hene$^{1}$\orcidlink{0009-0002-0336-374X}, 
    Kelin Luo$^{2}$\orcidlink{0000-0003-2006-0601} \\
    \\
    $^{1}$Institute of Computer Science, University of Bonn, Bonn, Germany \\
    \texttt{\{bartlmae, hene\}@uni-bonn.de} \\
    \\
    $^{2}$Department of Computer Science and Engineering, University at Buffalo, NY, USA \\
    \texttt{kelinluo@buffalo.edu}
}

\date{}

\begin{document}

\maketitle

\begin{abstract} Fast shipping and efficient routing are key problems of modern logistics. Building on previous studies that address package delivery from a source node to a destination within a graph using multiple agents (such as vehicles, drones, and ships), we investigate the complexity of this problem in specialized graphs and with restricted agent types, both with and without predefined initial positions.    
Particularly, in this paper, we aim to minimize the delivery time for delivering a package. To achieve this, we utilize a set of collaborative agents, each capable of traversing a specific subset of the graph and operating at varying speeds. This challenge is encapsulated in the recently introduced Drone Delivery Problem with respect to delivery time (DDT). 
    
In this work, we show that the DDT  with predefined initial positions on a line is NP-hard, even when considering only agents with two distinct speeds. This refines the results presented by Erlebach, et al.~\cite{erlebach:drones}, who demonstrated the NP-hardness of DDT on a line with agents of arbitrary speeds. Additionally, we examine DDT in grid graphs without predefined initial positions, where each drone can freely choose its starting position. We show that the problem is NP-hard to approximate within a factor of $O(n^{1-\varepsilon}$), where $n$ is the size of the grid, even when all agents are restricted to two different speeds as well as rectangular movement areas. We conclude by providing an easy $O(n)$ approximation algorithm. 
\end{abstract}

\section{Introduction}
With the rapid growth in e-commerce it becomes increasingly more important for logistics companies to deliver packages to customers as fast as possible. 
The Prime Air project, launched by 
Amazon~\cite{aboutamazonDroneProgram}, introduces a new approach to parcel delivery. The German logistics company DHL 
has tested the delivery of medical supplies, such as blood samples, using flying mobile agents. 
When it comes to areas which are not connected to a well-formed infrastructure or areas threatened by war, delivery of crucial goods as medicine and food can help people in need~\cite{bamburry2015drones,hii2019evaluation}. 

The collaborative delivery aspect of the drone delivery problem, which involves multiple agents on a single delivery path, primarily arises from variations in agent speed~\cite{bartschi2018}, consumption rates~\cite{bartschi_et_al:LIPIcs.STACS.2017.10}, or constraints such as battery power~\cite{bartschi2020collaborative} and movement areas~\cite{erlebach:drones}.  
In this paper, we further refine the model introduced by Erlebach et al.~\cite{erlebach:drones}. When it comes to their model, each agent's movement is confined to its designated area. Specifically, agents are allowed unrestricted movement within their respective subgraphs. This setting is motivated by realistic drone delivery scenarios where airspace is regulated by licenses, restricting certain drones to specific regions. Additionally, it is a natural observation that different types of agents, such as ships, drones, and trucks, are capable of traversing distinct parts of the graph. In addition, each drone operates at designated speeds and has its own consumption rate. Although both objectives – minimizing delivery time and reducing consumption – have been extensively studied, our research will exclusively focus on minimizing delivery time. We define this particular challenge as the Drone Delivery Problem with respect to time (DDT).   

In this work we will focus on paths and grid graphs. We refer to the fastest delivery on a line as DDT-Line (Drone Delivery on a Line) and the fastest delivery on a grid as DDT-Grid. Our work is primarily of theoretical interest and has the potential to lead to more generalized applications. Considering paths is theoretically significant, as they are the simplest form of graphs, and practically relevant, especially for fast last-mile delivery of critical goods, such as medical supplies. The grid graph setting is particularly applicable to densely populated areas such as large towns. 


\subsubsection*{Related Work.}
 The delivery problem has been extensively studied through various models, including the Shortest Path Problem~\cite{ortega2022shortest}, the Traveling Salesperson Problem (TSP)~\cite{gutin2006traveling}, the Vehicle Routing Problem~\cite{toth2002vehicle},
 and Dial-a-Ride Problem~\cite{cordeau2007dial}.
All these models assume that each request, whether a single node or a pair of nodes, is served by a single vehicle. 
 
Our study focuses on the collaborative delivery problem, involving multiple agents (servers, drones, etc.) in a single delivery process. This topic is inspired by the need for efficiency, aiming for faster or energy-efficient deliveries,   and accommodating constraints like battery limitations and movement restrictions. The primary motivation is to optimize delivery paths, reducing consumption. 
Bärtschi et al.~\cite{bartschi2018} first explored the delivery of a package from a source node to a target node with the aim of minimizing delivery time. They demonstrated that the problem can be solved in 
 time $O(k^2m+kn^2+\text{APSP})$ for a single package, where $n$ is the number of vertices, $m$ is the number of edges, $k$ is the number of agents and APSP represents the time required to compute all-pairs shortest paths in a graph with $n$ nodes and $m$ edges. 
 Carvalho et al. \cite{carvalho2021fast} improved the time complexity by developing an algorithm that runs in $O(kn \log n + km)$ time. Interestingly, they demonstrated that the problem escalates to NP-hardness with the addition of a second package.  To minimize energy consumption in package delivery, Bärtschi et al.~\cite{bartschi2018} developed a polynomial-time algorithm for delivering a single package, but demonstrated NP-hardness for multiple packages. They also investigated optimizing delivery time and energy for a single package. Lexicographically minimizing (time, energy) is polynomially feasible~\cite{bartschi2017energy}, but minimizing any combination of time and energy proves to be NP-hard~\cite{bartschi2018}.

Considering the realistic scenario where each agent has limited working time or energy, which restricts its total distance traveled, exploring a collaborative path utilizing multiple agents to deliver packages becomes particularly valuable~\cite{chalopin2014data,chalopin2014dataicalp,bartschi2020collaborative}.  Chalopin et al.~\cite{chalopin2014data} first demonstrated that the energy-constrained drone delivery problem is NP-hard in general graphs. Subsequently, Chalopin et al.~\cite{chalopin2014dataicalp} showed that this variant remains NP-hard even on a path graph. Bärtschi et al.~\cite{bartschi2020collaborative} found that a variant requiring each agent to return to its initial location is solvable in polynomial time for tree networks. 

Our study, as well as the research in~\cite{erlebach:drones}, diverges fundamentally from studies constrained by energy budgets. In our case, each agent's travel distance is not confined by a strict budget but is restricted to a specific subgraph where it can travel unlimited distances. As a result, the hardness results and algorithmic findings from studies with energy constraints do not directly translate to our problem. 
Erlebach et al.~\cite{erlebach:drones} first introduced this movement-restricted model. They demonstrate that for collaborative agents with restricted movement areas, it is NP-hard to approximate the Drone Delivery Time (DDT) within $O(n^{1-\varepsilon})$ or $O(k^{1-\varepsilon})$ 
 on general graphs, even if agents have equal speeds. 
For the case without initial positioning, they show that no polynomial time approximation with a finite ratio exists unless 
P$=$NP. Additionally, they argue that solving instances on a path remains NP-hard, when there are an arbitrary amount of different speeds. 


\subsubsection*{Our results.}

In section \ref{npline}, we provide stronger hardness results (see Theorem~\ref{thm:line}) than those presented by Erlebach et al.~\cite{erlebach:drones}. 
\begin{restatable}{theorem}
{thmline}\label{thm:line} 
The Drone Delivery Problem with initial positions on a line (DDT-Line) is NP-hard, even if all agents have only two different speeds.  
\end{restatable}

In Section \ref{npgrid}, we present a strong result for (unit) grid graphs:  

\begin{restatable}{theorem}
{thmgrid}\label{thm:grid_2speed}  
  For any constant $\varepsilon > 0$, the Drone Delivery Problem on a grid with agents following rectangular movement areas (DDT-GridR) without initial positions is NP-hard to approximate within a factor of $O(n^{1-\varepsilon})$, where $n$ denotes the size of the grid, even if all agents have only two different speeds.  
\end{restatable}


\section{Preliminaries}

We now formally define the Drone Delivery Problem with respect to time (DDT), following the notation used in Erlebach et al.~\cite{erlebach:drones}. To ensure completeness, we include a formal definition here.

 An instance of DDT is a tuple $I = (G, (s, y), A)$, where $G = (V, E, \ell)$ is an undirected graph, $\ell: E \rightarrow \mathbb{R}_{\geq 0}$ represents the edge lengths. The package  $(s, y)$ is defined by its starting location $s\in V$ and destination  $y\in V$, and $A$ is a set of $k$ drones.   
    \begin{enumerate}
    \item In DDT \textbf{with} predefined
initial positions, each drone $a \in A$ is represented as a tuple $a=(p_a,v_a,G_a)$, where $p_a$ is the initial starting position, $v_a$ is the agents speed, and $G_a=(V_a, E_a)$ is a connected subgraph of $G$ that defines the agent's movement area.  We assume that $p_a\in V_a$.  

Let $L_a(u, w)$ represent the length of the shortest path from $u$ to $w$ in agent $a$'s subgraph $G_a$. The time it takes for an agent $a$ to travel from $u$ to $w$ is given by  $\frac{L_a(u, w)}{v_a}$. 

    \item In DDT \textbf{without} predefined
initial positions, each drone $a \in A$ is represented as a tuple $a=(v_a,G_a)$, where  $v_a$ is the agents speed, and $G_a=(V_a, E_a)$ is a connected subgraph of $G$ that defines the agent's movement area. In contrast to DDT with predefined
initial positions, the initial position of each drone $a$ can be freely chosen within $V_a$.
    \end{enumerate}


A feasible solution to the DDT is a schedule of consecutive trips by agents carrying the package from $s$ to $y$. Each agent's trip consists of two phases: the \emph{empty phase}, during which the agent moves to its pickup location without the package, and the \emph{delivery phase}, during which the agent transports the package.  
In this paper, we assume that packages can be handed over between agents at vertices only. 
The objective in solving the DDT is to find a schedule that minimizes the total delivery time. In the section on hardness results, we further define the corresponding graph structures as either a path or a grid graph, and use DDT-Line and DDT-Grid to represent the respective problems. 

According to an observation in \cite{erlebach:drones} for any instance of DDT, there exists an optimal solution in which each involved agent picks up and drops off the package exactly once. Therefore, without loss of generality, we restrict our focus to instances where each drone participates in the delivery route at most once.  
We can now define a feasible solution - i.e. a schedule - by describing each agents' movement and the path the package takes. 
\begin{itemize}
\item 
The delivery route for each drone $a$ is represented by a 
triplet $(u, w, t)$, where agent $a$ picks up the package at time $t$ at location $u$ and drops it off at location $w$. Formally, each triplet must satisfy that $u, w \in V_a$.   
For the DDT with predefined initial positions, each involved drone $a$ must include an additional initial empty phase: $(p_a, u, 0)$, where $a$ moves from its initial position $p_a$ to $u$ at time $0$, in preparation for its delivery phase $(u, w, t)$. It must hold that $t\ge \frac{L_a(p_a, u)}{v_a} $ where $L_a(p_a, u)$ represents the length of the shortest path from $p_a$ to $u$ in the agent $a$'s subgraph $G_a$. 
In contrast, in DDT without predefined initial positions, 
the empty phase is not required; that is, the triplet must satisfy $t \ge 0$ and $u = p_a$, 
as we can directly select $u$ as the initial position. 

\item A feasible solution includes the package route in the form of tuples $(u, w, t, a)$, where the package moves with drone $a$ from $u$ to $w$ at time $t$. The first tuple must satisfy $u = s$, and the last must satisfy $w = y$.  
Formally, each tuple must satisfy that $u, w \in V_a$ and $(u, w, t) \in T_a$. 
Furthermore, for consecutive tuples $(u, w, t, a)$ and $(u', w', t', a')$, the condition $w = u'$ must hold, ensuring the route is continuous. Additionally, the timing constraint $t + \frac{L_a(u, w)}{v_a} \leq t'$ must be satisfied, ensuring that the subsequent drone picks up the package only after the previous drone drops it off.

\item  We define $t(S) = t + \frac{L_a(u, w)}{v_{a}}$ as the duration of this solution, where $(u, w, t, a)$ is the last tuple where $w= y$. The objective is to minimize $t(S)$.
\end{itemize}

In this paper, we study scenarios where the given $G$ represents either a line (path) graph or a grid graph, with details provided in the corresponding sections. Specifically, we consider DDT-Line with predefined initial positions and DDT-Grid without initial positions.

\section{Hardness result on a line}
\label{npline}


For the problem on a line, it is intuitive to conceptualize the drones' movement areas as a series of intervals. Specifically, for an agent $a$ with its movement area $V_a$ connected by edges $E_a$, we can represent $a$ as an interval $[u,u']$, where $u$ is the leftmost vertex of $V_a$  and $u'$ is the rightmost vertex. This interval notation helps simplify the representation of each agent’s range of movement on the line, facilitating easier analysis of their interactions and potential overlap within the path graph.  Figure~\ref{fig:DDTlineexample} gives an example.
    \begin{figure}
        \centering
        \includegraphics[scale=0.55]{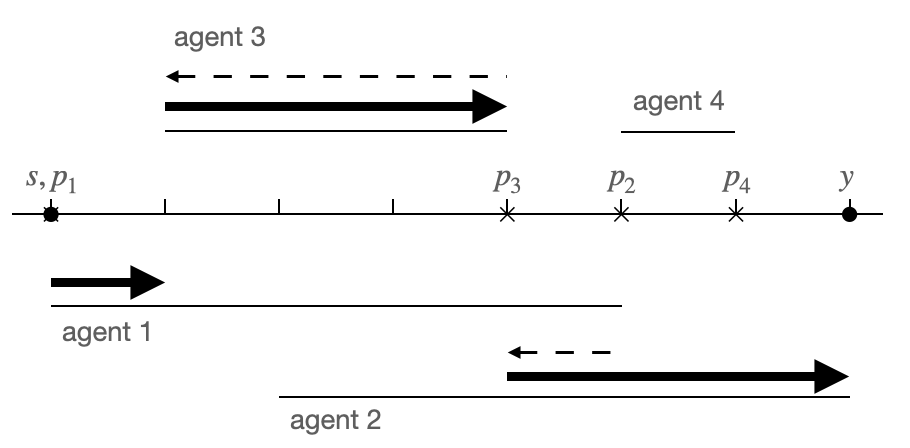}
        \caption{An example of a DDT-Line instance and its optimal solution is depicted: The solid line from $s$ to $y$ forms the path graph, where each solid edge has a distance of $1$. The initial positions of four drones are $p_1=s, p_2, p_3, p_4$, with the solid lines above and below representing each drone's movement subgraph. Suppose drone $1$ and $2$ have a speed of $1$, drone $3$ and drone $4$ have a speed of $3$. The optimal solution for this example takes a total time of $ \max\{1, 3/3\} + 3/3 + 3 = 5$ by sequentially using drones $1 $,  $ 3$ and $ 2$ to deliver the package from  $s$ to $y$, as indicated by the bold solid arrows. The dashed arrows represent the respective empty phases.} 
        \label{fig:DDTlineexample}
    \end{figure}

The hardness of the DDT problem on a line (DDT-Line) was first demonstrated in~\cite{erlebach:drones}. Their construction requires that all drones have varying speeds proportional to the input size.  Furthermore, they noted that the DDT-Line can be solved in polynomial time if all drones operate at the same speed, as all agents would simply move towards and follow the package. 
We advance this research by showing that the problem remains NP-hard even if we limit the scenario to agents with only two different speeds. 

\thmline*

We achieve this through a reduction from the well-studied NP-complete Partition problem~\cite{garey1979computers}, which provides a convenient structure for our proof. This reduction closes the gap between trivially solvable instances with unit speed, and the more complex scenarios previously established with $O(k)$ different speeds. First we discuss the construction and high level idea of the construction and later in the section we state a formal proof.

Let a set of positive integers $M= \{p_1,...,p_n\}$ be a partition problem instance with $ |M|= n$ and $P=\sum_{i\in [n]} p_i$. Without loss of generality, we assume that the numbers in $M$ are arranged in non-descending order, i.e.\ $p_i\leq p_j, \forall i<j$. The Partition problem involves determining whether there exists a partition $S\subset M$ such that $\sum_{S}p_i=\sum_{M\setminus S}p_i$.  In our DDT-Line instance, 
 we want each $p_i$ to be associated with an agent in DDT-Line and to offer two options: help on the left-side path or help on the right-side path. Then, we construct the scenario where the optimal (fastest) schedule is achieved if and only if the agents are arranged so that the sum of their associated integers 
on both sides equals exactly $\frac{P}{2}$. 

We now begin constructing an instance for the DDT-Line.    
To assist with the illustration, we have provided two figures:  
Figure \ref{fig:line_2speed} shows the layout of all drones' operating areas within a path graph, while Figure \ref{fig:line_2speed_gadget} provides a detailed view of the construction. Two different speeds are represented by colors in both Figure~s: red for the fast agents and black for the slow ones.    

For every $p_i$ from the Partition instance, we associate one slow agent $e_i$ (referred to as the \emph{element agent}), two fast agents $f_i^{l}$ and $f_i^{r}$, as well as two slow agents $b_i^l$ and $b_i^r$ (referred to as the \emph{base agents}) for the DDT-Line.  
For every $b_i^l$ (resp, $b_i^r$), there is a fast agent $h_i^l$ (resp, $h_i^r$) (referred to as the \emph{helping agents}).  For simplification, we use $b_i$ (resp. $ h_i $, $ f_i$) to denote either $b^l_i$ and  $b^r_i$ (resp, $h^l_i$ and  $h^r_i$, $f^l_i$ and  $f^r_i$). 
Additionally we have the slow auxiliary agents $d$ and $p$, as well as the fast auxiliary agent $q$. Agent $d$ serves as a \emph{delay}. Every feasible schedule has to assign the package to $d$ to initiate delivery from starting point $s$ and wait for it to traverse the interval of $d$. By selecting an appropriate length for $d$, we ensure that there is sufficient time for the $e_i$ agents to reach their designated helping spot (empty phase). Note that this \emph{delay} is a commonly employed technique in constructing hardness proofs for the DDT.   

\begin{figure}[ht]
    \centering
    \input{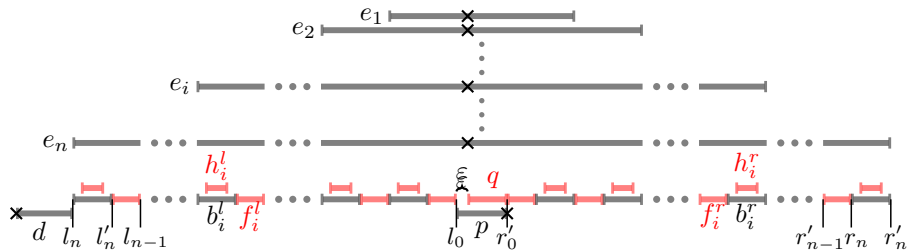}
    \caption{A sketch of the DDT-Line instance construction. There are $3n+2$ slow agents $\{e_i, b^l_i, b^r_i\}_{i\in [n]} \cup \{d, p\}$, $4n+1$ fast agents $\{f^l_i,f^r_i, h^l_i, h^r_i\}_{i\in [n]} \cup \{q\}$. The starting point $s$ is the leftmost node of agent $d$'s interval, and the destination $y$ is the rightmost point of agent $e_n$'s interval. In the middle part we have a gap of size $\varepsilon$ which can only be traversed by the $e_i$'s or $p$. Assigning $e_i$'s to the left side such that the respective $p_i$'s sum up to exactly $\frac{P}{2}$, the package meets $p$ at $l_0$ without waiting time.}
\label{fig:line_2speed}
\end{figure}

Observe that in the DDT-Line instance, we have alternating slow ($b_i$) and fast ($f_i$) base agents covering the main part of the $s-y$ interval. These agents cover the entire range except for a small gap of length  $\varepsilon$ in the middle, which is covered by the designated package carrier $p$ as well as the element agents. However, the instance is designed in a way such that we will have to choose $p$ here to achieve optimal delivery time. The $f_i$ agents are strategically designed so that no optimal schedule can bypass them. For the intervals corresponding to $b_i$, we have two options: either stick entirely with $b_i$ or use $h_i$ and a feasible $e_i$. The idea is that agent $h_i$ only benefits our schedule if we deliver the package to it and then pick it up from its right-most point by ``another agent''; otherwise, simply staying with $b_i $ yields the same solution. This introduces the role of $e_i$ as the ``another agent''. An agent $e_i$ acts as the counterpart to $b_i$ by either delivering to or picking up from $h_i$'s interval. Although any $e_j$ with $j>i$ could also play the respective counterpart or even stand in for $b_i$, this will not happen for any optimal solution, as we will discuss later. A detailed illustration of the $b_i$ and $h_i$ construction is provided in Figure  \ref{fig:line_2speed_gadget}.  

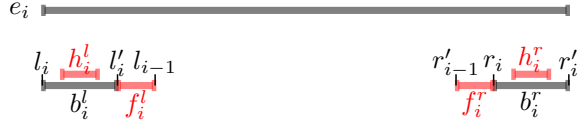
\begin{figure}[ht]
    \centering
    \begin{tikzpicture}[scale=0.5]
\definecolor{dgreen}{RGB}{13, 138, 8}
\draw[black, opacity=0.5, line width=2.7pt]   (11.06, 2) -- (24.94, 2);
\draw[black, opacity=0.5, line width=0.06*1cm/1pt]   (11.03, 2.15) -- (11.03, 1.85);
\draw[black, opacity=0.5, line width=0.06*1cm/1pt]   (24.97, 2.15) -- (24.97, 1.85);
\draw[black, opacity=0.5, line width=2.7pt]   (11.06, 0) -- (12.94, 0);
\draw[black, opacity=0.5, line width=0.06*1cm/1pt]   (11.03, 0.15) -- (11.03, -0.15);
\draw[black, opacity=0.5, line width=0.06*1cm/1pt]   (12.97, 0.15) -- (12.97, -0.15);
\draw[red, opacity=0.5, line width=2.7pt]   (11.56, 0.3) -- (12.44, 0.3);
\draw[red, opacity=0.5, line width=0.06*1cm/1pt]   (11.53, 0.44999999999999996) -- (11.53, 0.15);
\draw[red, opacity=0.5, line width=0.06*1cm/1pt]   (12.47, 0.44999999999999996) -- (12.47, 0.15);
\draw[red, opacity=0.5, line width=2.7pt]   (13.06, 0) -- (13.94, 0);
\draw[red, opacity=0.5, line width=0.06*1cm/1pt]   (13.03, 0.15) -- (13.03, -0.15);
\draw[red, opacity=0.5, line width=0.06*1cm/1pt]   (13.97, 0.15) -- (13.97, -0.15);
\draw[black, opacity=0.5, line width=2.7pt]   (23.06, 0) -- (24.94, 0);
\draw[black, opacity=0.5, line width=0.06*1cm/1pt]   (23.03, 0.15) -- (23.03, -0.15);
\draw[black, opacity=0.5, line width=0.06*1cm/1pt]   (24.97, 0.15) -- (24.97, -0.15);
\draw[red, opacity=0.5, line width=2.7pt]   (23.56, 0.3) -- (24.44, 0.3);
\draw[red, opacity=0.5, line width=0.06*1cm/1pt]   (23.53, 0.44999999999999996) -- (23.53, 0.15);
\draw[red, opacity=0.5, line width=0.06*1cm/1pt]   (24.47, 0.44999999999999996) -- (24.47, 0.15);
\draw[red, opacity=0.5, line width=2.7pt]   (22.06, 0) -- (22.94, 0);
\draw[red, opacity=0.5, line width=0.06*1cm/1pt]   (22.03, 0.15) -- (22.03, -0.15);
\draw[red, opacity=0.5, line width=0.06*1cm/1pt]   (22.97, 0.15) -- (22.97, -0.15);
\draw[black, line width=0.02*1cm/1pt]   (11, 0.3) -- (11, 0);
\draw[black, line width=0.02*1cm/1pt]   (13, 0.3) -- (13, 0);
\draw[black, line width=0.02*1cm/1pt]   (23, 0.3) -- (23, 0);
\draw[black, line width=0.02*1cm/1pt]   (25, 0.3) -- (25, 0);
\draw[black, line width=0.02*1cm/1pt]   (22, 0.3) -- (22, 0);
\draw[black, line width=0.02*1cm/1pt]   (14, 0.3) -- (14, 0);
\node[] at (12, -0.5) {\color{black}$b_i^l$};
\node[] at (13.5, -0.5) {\color{red}$f_i^l$};
\node[] at (24, -0.5) {\color{black}$b_i^r$};
\node[] at (22.5, -0.5) {\color{red}$f_i^r$};
\node[] at (10.4, 2) {\color{black}$e_i$};
\node[] at (12, 0.7) {\color{red}$h_i^l$};
\node[] at (24, 0.7) {\color{red}$h_i^r$};
\node[] at (14, 0.6) {\color{black}$l_{i-1}$};
\node[] at (11, 0.6) {\color{black}$l_i$};
\node[] at (13, 0.6) {\color{black}$l_i'$};
\node[] at (23, 0.6) {\color{black}$r_i$};
\node[] at (22, 0.6) {\color{black}$r_{i-1}'$};
\node[] at (25, 0.6) {\color{black}$r_i'$};
\end{tikzpicture}
    \caption{A close-up on $p_i$'s corresponding agents: slow agent $e_i$, $b_i^l $ and $ b_i^r$;  fast agents $f_i^l $, $ f_i^r$, $h_i^l$, and $ h_i^r$. From the construction, it is clear that  $e_i$ can only help only on either the left or right side interval, utilizing the fast drones $h_i$ on that side. A helping agent $h_i$ is utilized by assigning $e_i$ to either pick up the package from or deliver it to $h_i$. Not assigning $e_i$ leaves $h_i$ ineffective, as the delivery time is only dependent on $b_i$ as it has to catch up and pick up the package again.}
\label{fig:line_2speed_gadget}
\end{figure}

Building on the above observation, we will briefly outline the design goals for these intervals. 
It is important to note that the interval length of each agent $b_i$ depends on the underlying $p_i$, allowing us to construct both sides of the partition instance effectively. This means that the interval length of $b_i$ grows with the size $p_i$.  On the boundaries of the instance, the $b_i$'s have the largest lengths while they decrease as we move to the middle. It is clear that it benefits the delivery time the most if an element agent helps a larger $b_i$. To prevent any $e_i$ from helping any $b_j$ with $j>i$, we previously ordered the $p_i$ (and correspondingly the $e_i$) in increasing sequence and ensure that all $b_j$ with $j>i$ are inaccessible to $e_i$ (see Figure~\ref{fig:line_2speed}). As a result, $e_i$ can only help $b_j$ with $j\leq i$, while the optimal helping spot is then $b_i$ (either $b^l_i$ or  $b^r_i$). If $e_i$ would help at some $b_j$ with $j<i$ the schedule could be improved through an exchange argument. As large parts of the instance can be traversed with fast speed ($f_i$'s), it will not happen for any $e_i$ to help multiple $b_j$'s. This is exactly why there has to be made a choice between left and right by every $e_i$ to help a designated $b_j$ ($j\leq i$) which is $b_i^r$ or $b_i^l$ optimally.

To determine if the Partition instance is a ``yes''-instance, we use agent $p$. This agent is strategically positioned so that the package reaches agent $p$'s leftmost border precisely when the underlying $p_i$'s of the helping $e_i$'s sum up to exactly $\frac{P}{2}$ on the left side. Agent $p$ is positioned at its right border and upon initiating the delivery it starts moving to the left. Therefore we have to choose the length of $p$ to coincide with the exact time required to traverse the left side of the instance (including $d$) assuming we assign the $e_i$ agents corresponding to a feasible partition.
Only $p$ and all $e_i$ agents can traverse the middle part of our instance. The idea is that if we can not meet agent $p$ at the perfect time or use some agent $e_i$ instead of agent $p$ to traverse the middle gap, optimal delivery time cannot be achieved. 
Note that $p$ also cannot help any other $b_i^l$ or $b_i^r$ agents as there is no overlap between them.
There is another nuance to consider when it comes to the interval lengths of the agents $\{b^l_i\}_{i\in [n]}$. Assume for a moment that $b_i^l$ and $b_i^r$
were of equal length, we could achieve an optimal schedule without accurate partitioning – simply by assigning all agents to the right side, gaining the same benefits from helping and meeting $p$ on its left border after it has waited for the package. This strategy would result in the delivery time as if we had implemented a correct partition on both sides. To address this issue, we make  $b_i^l$ slightly larger by a factor of  $1+\frac{1}{P}$ to incentivize helping $b^l_i$. 
Then it is clear, that we lose optimal delivery time by assigning more agents to the right side as slower agents cover larger intervals. On the other hand, assigning agents corresponding to more than $\frac{P}{2}$ to the right side results in the package having to wait for $p$ and thus also losing the optimal delivery time, therefore the optimal schedule assigns the agents corresponding to a feasible partition, if possible. For completeness, we now state a formal proof of Theorem \ref{thm:line} and give the precise construction.

\paragraph*{Proof of Theorem 1.}
\label{sec:appendixDDTL}
First, we formally define the instance, beginning with variables corresponding to the leftmost and rightmost points of the intervals: $l_i $ and $ l_i'$ for the left side of agent $p$'s interval, $r_i$ and $r_i'$ for right of agent $p$'s interval, respectively (see Figure \ref{fig:line_2speed} and \ref{fig:line_2speed_gadget}). For $l_n$, we define it as $l_n=(2n+2)P^2-\frac{P}{2}-\frac{n}{2}-\frac{1}{2P}-1 + \varepsilon$, which represents the length of delay agent $d$ (i.e.\ the offset from $0$). We can write the following recursions for the left side:
\begin{align*}
    \forall i\in\{1,...,n\}:\, l_i'=l_i+(1+\frac{1}{P})p_i \\
    \forall i\in\{0,...,n-1\}:\, l_i=l_{i+1}'+P
\end{align*}
Correspondingly for the right side we can write a similar recursion. Serving as the base, $r_0'=l_0+(2n+2)P^2+\varepsilon$. It follows that
\begin{align*}
    \forall i\in\{1,...,n\}: r_i'=r_i+p_i \\
    \forall i\in\{1,...,n\}: r_i=r_{i-1}'+P.
\end{align*}
We can now define all agents' movement intervals using the above specified variables. Figure \ref{fig:line_2speed} also illustrates the structure of these variables. The available speeds are $1$ and $2P$. We define agent set $A$  to include: 
\begin{itemize}
    \item[$\bullet$] $d=(0,1,[0,l_n])$
    \item[$\bullet$] $p$ = $(r_0',1,[l_0,r_0'])$ 
    \item[$\bullet$] $q$ = $(r_0',2P,[l_0+\varepsilon,r_0'])$
    \item[$\bullet$] $\forall i\{1,...,n\}: b_i^l=(l_i,1,[l_i,l_i'])$
    \item[$\bullet$] $\forall i\{1,...,n\}: f_i^l=(l_i',2P,[l_i',l_{i-1}])$
    \item[$\bullet$] $\forall i\{1,...,n\}: b_i^r=(r_i,1,[r_i,r_i'])$
    \item[$\bullet$] $\forall i\{1,...,n\}: f_i^r=(r_{i-1}',2P,[r_{i-1}',r_i])$
    \item[$\bullet$] $\forall i\{1,...,n\}: e_i = (l_0+\varepsilon,1,[l_i,r_i'])$
    \item[$\bullet$] $\forall i\{1,...,n\}: h_i^l=(l_i+\delta,2P,[l_i+\delta,l_i'-\delta])$ with $\delta=(1+\frac{1}{P})\frac{p_i}{4P-2}$
    \item[$\bullet$] $\forall i\{1,...,n\}: h_i^r=(r_i+\delta',2P,[r_i+\delta',r_i'-\delta']$ with $\delta'=\frac{p_i}{4P-2}$
\end{itemize}

The value $\delta$ ($\delta'$ resp.) is chosen in a way such that the delivery time over the interval corresponding to $b_i^l$ ($b_i^r$) is $(1+\frac{1}{P})p_i$ ($p_i$) if there is no help and $(1+\frac{1}{P})\frac{p_i}{P}$ ($\frac{p_i}{P})$ otherwise. Essentially, a factor of $\frac{1}{P}$ is saved over the interval $b_i$ by assigning a suitable $e_j$ to help – optimally $e_i$ for the fastest overall schedule.

Consider the package to be $(0,r_n')$. 
We want to prove that the input of Partition is a ``yes''-instance if and only if there exists a schedule of time at most $t=(2n+2)P^2+(n+\frac{3}{2})P + \frac{n}{2} + \frac{1}{2} + 2\varepsilon$. Note that this is exactly the time a schedule needs if a feasible solution for Partition is assigned to each side of the DDT-Line instance. The threshold time $t$ exactly captures our previous observations: We assign element agents $e_i$ to both sides such that their underlying $p_i$'s sum up to exactly $\frac{P}{2}$; therefore we meet $p$ exactly on arrival at $l_0$ carrying the package over the gap and from there utilizing $q$ and proceeding with the right part of the instance.
 
On the other hand, if the underlying Partition instance is a ``no''-instance, then we can not beat the threshold time $t$. We argue that it is not feasible to skip agent $p$ or any $\{f_i^l, f_i^r\}_{i\in [n]}$. Assuming this is true, if we do not meet $p$ at the perfect time, the package either has to wait for agent $p$, or agent $p$ has to wait for the package.
In either case, achieving the desired threshold time is unattainable. In the first case, agents on the right cannot make up for the delay, and in the second case, we lost too much time on the left side as the intervals (corresponding to the $b_i^l$) are slightly larger. 

    
We will show two different lemmas to help proving our statement.

    \begin{lemma}
    \label{lemma:noskip}
        Any optimal schedule does not skip any agent $f_i^l$ or $f_i^r$ for all $i\in \{1, \dots, n\}$.
    \end{lemma}
    \begin{proof}
        Assume we want to skip any of the $f_i^l$. The only reason might be that we want to use the corresponding $e_i$ to pick up the package from $h_i^l$ and delivers it to $h_{i-1}^l$. That means $e_i$ carries the package over an additional distance $P$ and we therefore can get value from two helping agents using only one element agent. We assume we are at the left side of the instance as it is more beneficial regarding the helping agents. For the right side the same steps can be repeated.

        Let us compare the delivery times: First of all we have the time of the skipping strategy which is $t_{skip}=(1+\frac{1}{P})(\frac{p_i+p_{i-1}}{P})+P$ for some $i\geq 2$. The delivery time whenever we do not skip $f_i^l$ and not get helped with any of the adjacent intervals is $t_f=(1+\frac{1}{P})(p_i + p_{i-1}) + \frac{1}{2}$.

        It holds that $t_f<t_{skip}$, since

        \begin{align*}
                   (1+\frac{1}{P})(p_i+p_{i-1})+\frac{1}{2} &< (1+\frac{1}{P})(\frac{p_i+p_{i-1}}{P})+P\\
                    p_i + p_{i-1} + \frac{1}{2} &< \frac{p_i+p_{i-1}}{P^2} + \sum_{j\in[n]}p_j\\
                    \frac{1}{2} &< \frac{p_i+p_{i-1}}{P^2} + \sum_{j\in[n]\setminus\{i-1,i\}}p_j,
        \end{align*}
        where the last term is always true for any $n\geq 3$ as every input integer for Partition is at least 1. 
    \end{proof}

    We continue by arguing that is not possible to skip agent $p$. 

    \begin{lemma}
        \label{lemma:skipp}
        A schedule skipping $p$ does not deliver the package within time $t$.
    \end{lemma}

    \begin{proof}
     If we want to skip agent $p$ we have to assign some element agent (say $e_1$ as it has the least impact) to carry the package over the gap. However, if we do not use $p$ anyway we might as well get all gains on the left side as they are more beneficial by a factor $1+\frac{1}{P}$. The resulting strategy (call it the \emph{greedy strategy}) is the fastest among those skipping $p$. Lemma \ref{lemma:noskip} implies that all $f_i$ are part of our solutions. Therefore the proposed strategy differs only in the assignment of the element agents. The delay, the $f_i$ agents, as well as the middle part (agent $q$ as well as $p$ or $e_1$) have to be present in all schedules.
     This boils down the differences in delivery time to $t_{greedy}= (1+\frac{1}{P})(\frac{\sum_{2}^n p_j}{P}+p_1)+P$ for the greedy strategy. The first term result from the greedy left part and the second part represents the right part. For the schedule utilizing a perfect partition we have that $t^*=(1+\frac{1}{P})(\frac{1}{2}+\frac{P}{2})+ \frac{1}{2}+\frac{P}{2}$, where once again the first term represents the left side  and the other terms represent the right side of the instance. We want to show that $t^*<t_{greedy}$. Observe that $t^*=P + \frac{3}{2} + \frac{1}{2P}$ and $t_{greedy}= P + p_1 + 1 + \frac{1}{P} - \frac{p_1}{P^2}$. Together we get that

     \begin{align*}
        \frac{1}{2} + \frac{1}{2P} < p_1 + \frac{1}{P} - \frac{p_1}{P^2},
     \end{align*}
     which is true since $p_1\geq 1$ and $\frac{p_1}{P^2} \leq \frac{1}{2}$ for $n\geq 2$, proving the lemma.      
    \end{proof}

    Lemma \ref{lemma:noskip} together with Lemma \ref{lemma:skipp} imply that a schedule beating time $t$ needs to use $p$ to carry the package over the gap and assign each element agent $e_i$ to exactly one base agent $b_i$. With the instance construction and lemmas established, we are now ready to prove the following. 

\thmline* 

\begin{proof}
    Assume the input $M$ of Partition is a ``no''-instance, that is, there exists no subset $S\subset M$ with $\sum_{i\in S} p_i = \frac{P}{2}$. This implies that we can not meet $p$ at the perfect time. Either we allocate too little to the left side or too much. In the first case $p$ has to wait at the gap for the package. The gain that is made on the right side is smaller (by a factor of $\frac{1}{P}$) than a perfect partition could have achieved on the left side. Therefore we are too slow in this case. On the other hand, allocating too much to the left results in the package having to wait for $p$ to arrive. Thus the left side is just as fast as if we put a perfect partition. Since a perfect partition has more capacities on the right side, also overshooting on the left side results in a schedule that is too slow.

    There is another nuance to mention. So far we did not show that every agent can reach its desired helping spot in time. This case is especially crucial for agent $e_n$ if it wants to help on the right side. Agent $e_n$, as all element agents, starts at $l_0+\varepsilon$ and has to travel over the stretch (agent $q$). We demonstrate that $e_n$ has enough time to reach $r_n$ (and also $l_n$).

Assume that the input for Partition is a ``yes''-instance. Observe that if $e_n$ starts moving to the right as soon as possible it will be at $r_0'+\varepsilon$ whenever $p$ picks up the package at $l_0$. To reach $r_n$ agent $e_n$ has to travel distance $P-p_n+nP$ taking an equal amount of time $t_{e_n}$. We proceed with a pessimistic analysis: Assume for every $b_i^r$ with $1\leq i<n$ that it is helped. Then the package needs time $t_{package}=\frac{(2n+2)P^2}{2P} +\frac{nP}{2P}+1- \frac{p_n}{P}$, where the first term represents the stretch (agent $q$), the second term all $f_i^r$ and the last two terms the helped base agents. It holds that $t_{e_n}<t_{package}$ since

\begin{align*}
    (n+1)P - p_n &< \frac{(2n+2)P^2}{2P} +\frac{nP}{2P}+1- \frac{p_n}{P}\\
    (n+1)P - p_n &<(n+1)P + \frac{n}{2} + 1 - \frac{p_n}{P}\\
    - p_n &< \frac{n}{2} + 1 - \frac{p_n}{P}
\end{align*}
is true. 

In a similar fashion we can state that $e_n$ has enough time to reach $l_n$, i.e.\ the helping spot on the left side. Starting from $l_0 +\varepsilon$ it takes agent $e_n$ time $nP+(1+\frac{1}{P})P+\varepsilon$ to reach $l_n$. Due to the delay the package arrives at $l_n$ at time $(2n+2)P^2-\frac{P}{2}-1-\frac{1}{2P} + \varepsilon$, which is a sufficient amount of time for $n\geq 2$. As a consequence, $e_n$ has enough time to reach $l_n$. Together with the previous analysis we argued that all agents have enough time to reach their helping spot.

All in all we argued that there can not exist a schedule beating time $t$ whenever the underlying Partition input does not admit a perfect partition. This concludes the proof of Theorem \ref{thm:line}.
\end{proof}
It is said that we presented only a set of feasible values. There exist many more feasible values for speeds and distances that serve our construction.
\section{Hardness result on a grid}
\label{npgrid}
%

In this section, we study the DDT problem on grid graphs. The grid graph serves as the natural intermediary between a line and a general graph, and they are more closely aligned with applications such as road networks.  Formally, we define a grid graph as follows: The set of vertices is a finite set of integer coordinates $V \subset \mathbb{Z}^2$. 
An edge of length $1$ connects any two vertices if and only if exactly one of their coordinates differs by exactly one. Therefore we also refer to our grid graph as a \emph{unit grid}. We denote the DDT problem on the unit grid as DDT-Grid. 


\begin{figure}[ht]
    \centering
    \begin{tikzpicture}[scale=0.9]
\definecolor{b1}{RGB}{0, 0, 0}
\definecolor{b2}{RGB}{4, 75, 145}
\definecolor{b3}{RGB}{76, 68, 194}
\definecolor{b4}{RGB}{50, 140, 190}
\draw[->, line width=1.5] (0, -0.35) -- (3, -0.35);
\draw[<-, red, line width=1.5] (3.3, 3) -- (3.3, 0);
\draw[->, line width=1.5] (3, 3.5000000000000003) -- (1, 3.5000000000000003);
\draw[->, line width=1.5] (0.9, 3.35) -- (-0.3, 3.35) -- (-0.3, 2);
\def\wi{0.08}
\def\op{0.2}
\def\pts{2.0pt}

\foreach \y in {-1, ...,4} {
\draw[black, line width=\wi, opacity=\op] (-1.3, \y) -- (4.3, \y);
\draw[black, line width=\wi, opacity=\op] (5.7, \y) -- (11.3, \y);
}
\foreach \x in {-1, 0, 1, 2, 3, 4, 6, 7, 8, 9, 10, 11} {
\draw[black, line width=\wi, opacity=\op] (\x, -1.3) -- (\x, 4.3);
}


\foreach \x in {-1,...,4} {
\foreach \y in {-1,...,4} {
\fill[black, opacity=0.6] (\x, \y) circle (1.5pt);
}
}
\foreach \x in {6,...,11} {
\foreach \y in {-1,...,4} {
\fill[black, opacity=0.6] (\x, \y) circle (1.5pt);
}
}
\draw[b1, line width=1pt, opacity=0.8, rounded corners=3pt] (-0.12, -0.12) rectangle (3.12, 0.12);
\fill[b1, opacity=0.25, rounded corners=5pt] (-0.12, -0.12) rectangle (3.12, 0.12);

\draw[b1, line width=1pt, opacity=0.8, rounded corners=3pt] (1.88, -0.12) rectangle (2.12, 3.12);
\fill[b1, opacity=0.25, rounded corners=5pt] (1.88, -0.12) rectangle (2.12, 3.12);
\draw[b1, line width=1pt, opacity=0.8, rounded corners=3pt] (0.88, 2.88) rectangle (3.12, 3.12);
\fill[b1, opacity=0.25, rounded corners=5pt] (0.88, 2.88) rectangle (3.12, 3.12);
\draw[red, line width=1pt, opacity=0.8, rounded corners=3pt] (2.88, -0.12) rectangle (3.12, 3.12);
\fill[red, opacity=0.25, rounded corners=5pt] (2.88, -0.12) rectangle (3.12, 3.12);
\draw[b1, line width=1pt, opacity=0.8, rounded corners=3pt] (-0.12, 1.88) rectangle (1.12, 3.12);
\fill[b1, opacity=0.25, rounded corners=3pt] (-0.12, 1.88) rectangle (1.12, 3.12);
\draw[black, line width=1pt, rounded corners=3pt](7.88, 3.12) -- (7.88, 0.88) -- (10.120000000000001, 0.88) -- (10.120000000000001, 3.12) -- (8.88, 3.12) -- (8.88, 2.88) -- (9.879999999999999, 2.88) -- (9.879999999999999, 2.12) -- (8.88, 2.12) -- (8.88, 1.88) -- (9.879999999999999, 1.88) -- (9.879999999999999, 1.12) -- (8.12, 1.12) -- (8.12, 3.12) -- cycle;
\fill[black, opacity=0.25, rounded corners=5pt](7.88, 3.12) -- (7.88, 0.88) -- (10.120000000000001, 0.88) -- (10.120000000000001, 3.12) -- (8.88, 3.12) -- (8.88, 2.88) -- (9.879999999999999, 2.88) -- (9.879999999999999, 2.12) -- (8.88, 2.12) -- (8.88, 1.88) -- (9.879999999999999, 1.88) -- (9.879999999999999, 1.12) -- (8.12, 1.12) -- (8.12, 3.12) -- cycle;
\draw[black, line width=1pt, rounded corners=3pt](6.88, -0.12) -- (6.88, 3.12) -- (7.12, 3.12) -- (7.12, 2.12) -- (8.120000000000001, 2.12) -- (8.120000000000001, 1.88) -- (7.12, 1.88) -- (7.12, -0.12) -- cycle;
\fill[black, opacity=0.25, rounded corners=5pt](6.88, -0.12) -- (6.88, 3.12) -- (7.12, 3.12) -- (7.12, 2.12) -- (8.120000000000001, 2.12) -- (8.120000000000001, 1.88) -- (7.12, 1.88) -- (7.12, -0.12) -- cycle;
\draw[->, line width=1.5] (7.25, 0) -- (7.25, 1.75) -- (7.75, 1.75);
\draw[->, line width=1.5] (8.25, 1.9) -- (8.25, 1.25) -- (9.75, 1.25) -- (9.75, 1.75) -- (9, 1.75);
\filldraw (7, 0) circle (2pt);
\node[below left] at (7, 0) {$s$};
\filldraw (0, 2) circle (2pt);
\node[below left] at (0, 2) {$y$};
\filldraw (9, 2) circle (2pt);
\node[left] at (9, 2) {$y$};
\filldraw (0, 0) circle (2pt);
\node[below left] at (0, 0) {$s$};
\node[] at (9, 0.5) {\large $a_2$};
\node[] at (1.5, 1.5) {\large $a_2$};
\node[] at (1.5, -0.7) {\large $a_1$};
\node[] at (6.5, 1.5) {\large $a_1$};
\node[red] at (3.7, 1.5) {\large $a_5$};
\node[] at (-0.6, 2.8) {\large $a_3$};
\node[] at (2, 3.7) {\large $a_4$};
\node[above] at (1.5, -2) {\big (a)};
\node[above] at (8.5, -2) {\big (b)};
\end{tikzpicture}
    \caption{Two DDT-Grid instances. On the left is an instance (a) where agents have rectangular movement areas and have two distinct speeds: 
    there are four slow agents with speed $1$ and one fast agent with speed $5$, displayed black and red respectively. Each movement area is represented as a rectangle, indicated by shading. 
   The optimal solution, with the respective trips indicated by bold solid arrows, takes a total time of $7.6$ by sequentially using agents $a_1 $, $ a_5$, $a_4$ and $ a_3$ to deliver the package from $s$ to $y$. 
   On the right is an instance (b) with two agents having unit speed, where the movement areas of the agents are not rectangular.
    Note that even though the nodes of the subgraph of the second agent have rectangular shape, it is not a rectangular movement area as some edges are missing. The optimal solution takes a total time of $8$ by first using agent $a_1$ and then $a_2$ as indicated by the bold solid arrows.}
    \label{fig:rectangular_example}
\end{figure}
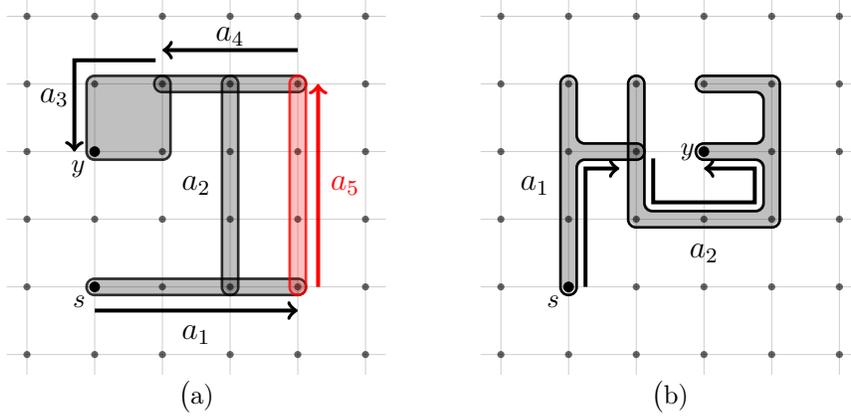

We will study the case of DDT-Grid involving agents that are restricted to rectangular movement areas. 
This restriction to rectangular movement areas is natural, reflecting real-world scenarios where agents' movements are typically constrained to specific, license-determined, convex areas, such as road networks.  
Let us define what constitutes a \emph{rectangular} movement area: For every agent $a$, the vertex set $V_a$ forms a rectangular shape, i.e., for every two vertices $x=(x_1,x_2) \in V_a$ and $y=(y_1,y_2)\in V_a$ with $x_1\leq y_1$ and $x_2\leq y_2$, every vertex $z=(z_1,z_2) \in \mathbb{Z}^2$ with $z_1 \in [x_1,y_1]$ and $z_2\in [x_2,y_2]$ must also belong to $V_a$. Regarding $E_a$, for every pair of vertices in $V_a$ differing by exactly 1 in one coordinate, the connecting edge $e$ must be in $E_a$. Intuitively, this ensures that the subgraph forms a rectangle where every possible vertex and edge within that rectangle is included.  This specific scenario will be referred to as DDT-GridR.  
An example of rectangular as well as arbitrary movement areas is given in Figure \ref{fig:rectangular_example}.  

If all agents have the same speed, we can easily provide a polynomial-time algorithm to solve the above DDT-GridR problem without initial positions by constructing the shortest path and generating the corresponding collaborative schedule. However, when agents operate at two different speeds, the problem becomes significantly more complex and challenging.  
We now state the main result regarding this setting. Notice that now unlike in section \ref{npline} we consider the setting of \emph{no initial positioning}. This means we have the flexibility to choose the initial placement of an agent rather than being constrained to fixed starting positions. We refer to the number of vertices $n$ in the grid graph as the \textit{size} of the grid.

\thmgrid*



 We prove the theorem by showing a reduction from \textsc{2P1N-3SAT}, a special case of \textsc{3SAT}, where each variable in the input formula appears exactly two times as a positive and one time as negative literal. 
It is known, that \textsc{2P1N-3SAT} is NP-hard \cite{ryo:2p1nsat}. 
As in Section~\ref{npline}, we will provide a high-level overview of the construction first and later in this section give a formal proof.


Let $\phi$, with $n'$ variables and $m'$ clauses, be an instance of the \textsc{2P1N-3SAT} problem. 
The core idea behind the construction of the reduction instance is as follows: For every (of the three) occurrences of each variable we construct an agent which has to make a choice – help with solving the clauses (clause gadgets) or do not help (variable gadgets). Intuitively, the first refers to the literal being set to true, while the latter refers to the literal being set to false in the corresponding assignment.
The design goal of the variable gadgets is to ensure that every assignment is consistent, i.e.\ if $x_i$ is set to true then $\neg x_i$ is set to false. The idea of the clause gadgets is to make sure that the assignment fulfills every clause of the formula. We obtain that the optimal delivery time is below a certain threshold if and only if the given input $\phi$ of \textsc{2P1N-3SAT} is a ``yes''-instance. Furthermore, the DDT-GridR instance is constructed in a way, such that if $\phi$ is not satisfiable, then the additional delivery time increases drastically.

\begin{figure}[h!]
    \centering
    \input{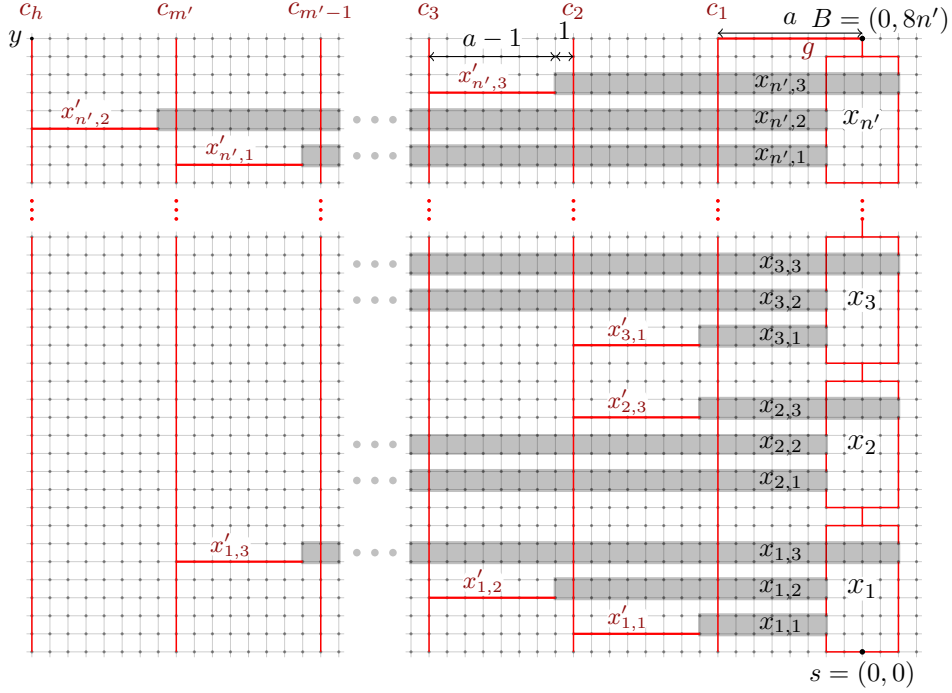}
    \caption{A sketch of the DDT-GridR instance: The package is to be delivered from the bottom-right to the top-left.  The grey bars represent the rectangular movement areas of the literal agents associated with the three occurrences of every variable (two positive $x_{i,1},x_{i,2}$, as well as one negative $x_{i,3}$ ). On the left side we have $m'$ clause gadgets. In this example clause $c_1$ contains literals $x_1$, $\neg x_2$, and $x_3$. Therefore we have the three agents $x'_{1,1},x'_{2,3},x'_{3,1}$, which work as counterparts to $x_{1,1},x_{2,3},x_{3,1}$ delivering the package between two clause gadgets. Agent $c_h$ (on the far left)  serves as an auxiliary agent to assist with the delivery to $y$. On the right side we have $n'$ variable gadgets. Every optimal schedule (given $\phi$ is satisfiable) must traverse all variable gadgets up to point $B$, as delivering horizontally over a distance $a$ with a slow agent would require an excessive amount of time.}
\label{fig:grid_2speed}
\end{figure}

Let us begin to construct the resulting instance for DDT-GridR. For a complete overview see Figure \ref{fig:grid_2speed}. The two different speeds are depicted in red (fast) and grey (slow). Note that every single straight red line represents a fast agent and we want to deliver the package from bottom right ($s$) to top left ($y$). 
For every occurrence of each variable $x_i$ we construct agents $x_{i,1},x_{i,2}$ and $x_{i,3}$ where each of the first two corresponds to one of the positive occurrences and the latter one corresponds to the negative occurrence $\neg x_i$. We will refer to these as the \textit{literal agents}. 
As briefly mentioned before, all literal agents referring to $x_i$ can now contribute to our schedule either in a clause gadget containing $x_i$ or $\neg x_i$, or in the variable gadget corresponding to $x_i$. More precisely, each of $x_{i,1}$ and $x_{i,2}$ have a distinct clause gadget assigned such that the clause contains literal $x_i$ and $x_{i,3}$ has the clause gadget assigned that has the unique occurrence of $\neg x_i$. From Figure \ref{fig:grid_2speed} we observe that the presence of a (slow) literal agent $x_{i,k}$ in clause $c_j$ is required to fill the gap of length 1 (left of the fast agent labeled with $c_j$) to deliver to the fast agent $x'_{i,k}$. We can think of filling this gap as satisfying the particular clause. 

We have discussed the inclusion of both the clause gadget and the variable gadget in the instance construction. First, let us take a closer look at the variable gadgets depicted on the right side of Figure \ref{fig:grid_2speed}. The design goal of the variable gadget is to ensure feasible assignments, that is, either $x_i$ or $\neg x_i$ – but not both – contributes to solving the clauses and is
set to true. A close-up on a variable gadget is given by Figure \ref{fig:grid_var_gadget}. Note that this is just one of all $n'$ consecutive variable gadgets and each gadget consists of multiple fast agents (every straight red line). For the traversal of a variable gadget from bottom to top, either both positive literal agents or the negative literal agent have to engage with it. If the traversal starts at $s$, passes through all consecutive variable gadgets, and finishes at
$B$ (see Fig~\ref{fig:grid_2speed}), we can infer that literal agents were placed to help carry the package across the gaps (two gaps on the left side or one gap on the right side of each variable gadget). 


As we move toward the clause gadgets (utilizing auxiliary agent $g$) we run into the following situation: To assist in traversing the gaps induced by the clause gadgets, we can only incorporate the  
literal agents that were not placed in the variable gadget. If we rely on a literal agent that was already placed on the right side of the instance, helping in the variable gadget, to also carry over a gap of a clause gadget,
we know that this agent has traveled at least a certain distance, $a$, at a slow speed.
Setting $a$ to be sufficiently large, we can conclude that a schedule utilizing a literal agent in the variable gadget as well as a clause gadget is too slow to beat a threshold time $t$ which assumes that all gaps can be traversed by distinct literal agents. The same holds for using a literal agent in multiple clause gadgets as they also have a horizontal distance of $a$. It is especially important, that the speed of the fast agents depends on $a$ and $a$ itself is relying on $n'$ and $\varepsilon$. 
Essentially, setting $a$ sufficiently large ensures that traveling a distance of $a$ at a slow speed results in a delivery time that is more than  
$n^{1-\varepsilon}$ (with $n$ being the size of the grid) times as large as the delivery time for a satisfiable instance of \textsc{2P1N-3SAT} (using the satisfying assignment in the clause gadgets and the complements in the variable gadgets).

Finally, we get to the essence of the reduction proof argument: 
If there exists a feasible assignment ($\phi$ is a ``yes''-instance), then we can place the agents in a way such that no literal agent has to bridge more than a single gap and the optimal delivery time must be within the threshold. Consequently, the delivery time of an $n^{1-\varepsilon}$-approximation algorithm must be within $n^{1-\varepsilon}$ times the threshold; On the other hand, if we have a sufficiently fast delivery schedule, we can rebuild the assignment by looking at the literal agents placed in the variable gadgets and output (the inverse) as a feasible solution. If the delivery time of the approximation algorithm is less than $n^{1-\varepsilon}$ times the threshold, it is impossible for a slow agent to have traveled a horizontal distance of $a$, as this would take more than $n^{1-\varepsilon}$ times the threshold. This ensures a feasible and satisfying assignment, therefore $\phi$ is a ``yes''-instance. 


\begin{figure}[ht]
    \centering
    \begin{tikzpicture}[scale=0.50]
\def\wi{0.08}
\def\op{0.2}
\def\pts{2.0pt}
\foreach \x in {-6,...,3} {
\foreach \y in {-1,...,8} {
\fill[black, opacity=0.6] (\x, \y) circle (\pts);
}
}

\foreach \x in {-6,...,3} {
\draw[black, line width = \wi, opacity=\op] (\x, -1.3) -- (\x, 8.3);
}
\foreach \y in {-1,...,8} {
    \draw[black, line width = \wi, opacity=\op] (-6.3,  \y) -- (3.3, \y);

}
\draw[red, line width=1pt, opacity=0.8, rounded corners=2pt] (-2.2, -0.2) rectangle (2.2, 0.2);
\fill[red, opacity=0.25, rounded corners=5pt] (-2.2, -0.2) rectangle (2.2, 0.2);


\draw[red, line width=1pt, opacity=0.8, rounded corners=2pt] (-2.2, -0.2) rectangle (-1.8, 1.2);
\fill[red, opacity=0.25, rounded corners=5pt] (-2.2, -0.2) rectangle (-1.8, 1.2);


\draw[red, line width=1pt, opacity=0.8, rounded corners=2pt] (-2.2, 1.8) rectangle (-1.8, 3.2);
\fill[red, opacity=0.25, rounded corners=5pt] (-2.2, 1.8) rectangle (-1.8, 3.2);


\draw[red, line width=1pt, opacity=0.8, rounded corners=2pt] (-2.2, 3.8) rectangle (-1.8, 7.2);
\fill[red, opacity=0.25, rounded corners=5pt] (-2.2, 3.8) rectangle (-1.8, 7.2);


\draw[red, line width=1pt, opacity=0.8, rounded corners=2pt] (1.8, -0.2) rectangle (2.2, 5.2);
\fill[red, opacity=0.25, rounded corners=5pt] (1.8, -0.2) rectangle (2.2, 5.2);


\draw[red, line width=1pt, opacity=0.8, rounded corners=2pt] (1.8, 5.8) rectangle (2.2, 7.2);
\fill[red, opacity=0.25, rounded corners=5pt] (1.8, 5.8) rectangle (2.2, 7.2);

\draw[red, line width=1pt, opacity=0.8, rounded corners=2pt] (-2.2, 6.8) rectangle (2.2, 7.2);
\fill[red, opacity=0.25, rounded corners=5pt] (-2.2, 6.8) rectangle (2.2, 7.2);

\draw[red, line width=1pt, opacity=0.8, rounded corners=2pt] (-0.2, 6.8) rectangle (0.2, 8.2);
\fill[red, opacity=0.25, rounded corners=5pt] (-0.2, 6.8) rectangle (0.2, 8.2);

\filldraw (0, 0) circle (3pt);
\node[below] at (0, 0) {start};
\filldraw (0, 8) circle (3pt);
\node[above] at (0, 8) {end};
\draw[black, line width = 1pt]   (-6.5, 2.1) -- (-1.87, 2.1);
\draw[black, line width = 1pt]   (-6.5, 0.9) -- (-1.87, 0.9);
\draw[black, line width = 1pt]   (-1.87, 2.1) -- (-1.87, 0.9);
\draw[black, line width = 1pt]   (-6.5, 4.1) -- (-1.87, 4.1);
\draw[black, line width = 1pt]   (-6.5, 2.9) -- (-1.87, 2.9);
\draw[black, line width = 1pt]   (-1.87, 4.1) -- (-1.87, 2.9);
\draw[black, line width = 1pt]   (-6.5, 6.1) -- (2.13, 6.1);
\draw[black, line width = 1pt]   (-6.5, 4.9) -- (2.13, 4.9);
\draw[black, line width = 1pt]   (2.13, 6.1) -- (2.13, 4.9);
\draw[black, opacity=0.25, line width = 15.590551185000002pt]   (-6.5, 1.5) -- (-1.87, 1.5);
\draw[black, opacity=0.25, line width = 15.590551185000002pt]   (-6.5, 3.5) -- (-1.87, 3.5);
\draw[black, opacity=0.25, line width = 15.590551185000002pt]   (-6.5, 5.5) -- (2.13, 5.5);
\fill[black, opacity=0.25] (-7, 1.5) circle (5pt);
\fill[black, opacity=0.25] (-7.5, 1.5) circle (5pt);
\fill[black, opacity=0.25] (-8, 1.5) circle (5pt);
\node at (-4, 1.5) { \large $x_{i, 1}$};
\fill[black, opacity=0.25] (-7, 3.5) circle (5pt);
\fill[black, opacity=0.25] (-7.5, 3.5) circle (5pt);
\fill[black, opacity=0.25] (-8, 3.5) circle (5pt);
\node at (-4, 3.5) { \large $x_{i, 2}$};
\fill[black, opacity=0.25] (-7, 5.5) circle (5pt);
\fill[black, opacity=0.25] (-7.5, 5.5) circle (5pt);
\fill[black, opacity=0.25] (-8, 5.5) circle (5pt);
\node at (-4, 5.5) { \large $x_{i, 3}$};
\end{tikzpicture}
    \caption{Depiction of the variable gadget of $x_i$. It consists of 8 fast agents (red). To deliver the package from ``start'' to ``end'', it must traverse either the left or right side of the gadget. For the left side, both $x_{i, 1}$ and $x_{i, 2}$ must be present to deliver the package across the two gaps; while for the right side, only $x_{i, 3}$ is required to cover a single gap. 
    Note that skipping the gadget and traversing the distance $a$ horizontally using a literal agent is too slow and therefore not feasible for satisfiable inputs.}
    \label{fig:grid_var_gadget}
\end{figure}
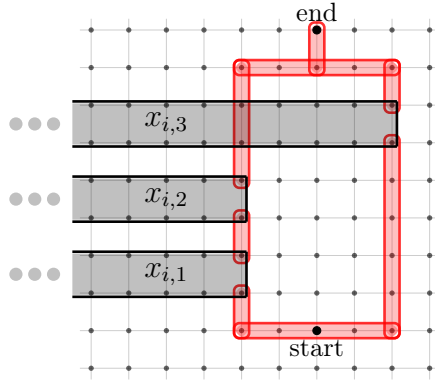

 
 We now state the DDT-GridR instance formally and provide a proof of Theorem \ref{thm:grid_2speed}.

\paragraph*{Proof of Theorem 2.}
\label{sec:appendixddtgrapx}


Let $\phi$ be the input formula of \textsc{2P1N-SAT} with $n'$ variables and $m'$ clauses. Observe that $m' \leq 3n'$.  
We will now formally state the DDT-GridR instance $I$, describing each agent by its speed and movement area. The speed of the fast agents is set to $a-1$ and the slow agent speed to $1$. We set the starting point of the package to be $s = (0,0)$ and the destination to be $y = (-(m'+1)\cdot a,8n')$ where $n'$ is the number of variables and $m'$ is the number of clauses in the \textsc{2P1N-3SAT} instance $\phi$ and $a$ is a large integer that is set in the proof, depending on $n'$ and $\varepsilon$. Note that $a$ is also the horizontal spacing between clause gadgets as well as the transition from variable to clause gadgets.

Notice that each variable gadget $i$ consist of $8$ fast agents see Fig \ref{fig:grid_var_gadget}) with the following rectangular movement areas: $\forall i \in \{1, \dots, n'\}: $ 
\begin{itemize}
    \item[$\bullet$] $(a-1, [-2, 2] \times [8\cdot(i-1), 8\cdot(i-1)])$
    \item[$\bullet$] $ (a-1, [-2, -2] \times [8\cdot(i-1), 8\cdot(i-1)+1])$
    \item[$\bullet$] $(a-1, [-2,-2] \times [ 8\cdot(i-1)+2, 8\cdot(i-1)+3])$ 
    \item[$\bullet$] $(a-1, [-2, -2] \times [8\cdot(i-1)+4, 8\cdot(i-1)+7])$ 
    \item[$\bullet$] $(a-1, [2, 2] \times [8\cdot(i-1), 8\cdot(i-1)+5])$
    \item[$\bullet$] $(a-1, [2, 2] \times [8\cdot(i-1)+6, 8\cdot(i-1)+7])$
    \item[$\bullet$] $(a-1, [-2, 2] \times [8\cdot(i-1)+7, 8\cdot(i-1)+7])$
    \item[$\bullet$] $(a-1, [0, 0] \times [8\cdot(i-1)+7, 8\cdot(i-1)+8])$
\end{itemize}
Recall that since we are in the setting of no initial positioning, the agents are triples as the initial position $p_a$ is not specified in the instance.

As defined, $x_{i, 1}$ and $x_{i, 2}$ represent the positive literals of $x_i$, and $x_{i, 3}$ the negative literal. For each literal $x_{i, k}$, we denote $c_{j_{i, k}}$ as the clause that includes the literal $x_{i, k}$. The remaining agents are identified by their respective names.

\begin{itemize}
\item 
Each clause gadget $j$ consist of a fast agents with the following rectangular movement area:  $\forall j \in \{1, \dots, m'\}: $ 

\begin{itemize}
    \item[$\bullet$] agent $c_j$:  $(a-1, [-j\cdot a, -j\cdot a] \times [0, 8n'])$
\end{itemize}
\item Additionally there is an auxiliary agent $c_h$: $(a-1, [-(m'+1)\cdot a, -(m'+1)\cdot a] \times [0, 8n'])$

\item  
For each variable $x_i$ we introduce the $3$ slow literal agents $x_{i, 1}, x_{i, 2}, x_{i, 3}$ and their respective counterparts $x'_{i,1},x'_{i, 2}, x'_{i, 3}$, depicted as gray bars and red lines, respectively in Figure \ref{fig:grid_2speed}. Formally they are defined using their speeds and movement areas:  
$\forall i \in \{1, \dots, n'\}: $ 
\begin{itemize}
    \item[$\bullet$] agent $ x_{i, 1}$:  $(1, [-j_{i, 1}\cdot a-1, -2] \times [8(i-1) + 1, 8(i-1) + 2])$
   \item[$\bullet$] agent $ x_{i, 2} $:  $ (1, [-j_{i, 2}\cdot a-1, -2]  \times [8(i-1) + 3, 8(i-1) + 4] )$
\item[$\bullet$] agent $ x_{i, 3} $:  $  (1, [-j_{i, 3}\cdot a-1, 2] \times [8(i-1) + 5, 8(i-1) + 6] )$
   \item[$\bullet$] agent $ x_{i, 1}'  $:  $ (a-1, [-(j_{i, 1}+1)\cdot a, -j_{i, 1}\cdot a-1]  \times [8(i-1) + 1, 8(i-1) + 1])$
  \item[$\bullet$] agent $ x_{i, 2}' $:  $ (a-1, [-(j_{i, 2}+1)\cdot a, -j_{i, 2}\cdot a-1] \times [8(i-1) + 3, 8(i-1) + 3]) $
 \item[$\bullet$] agent $ x_{i, 3}' $:  $(a-1, [-(j_{i, 3}+1)\cdot a, -j_{i, 3}\cdot a-1]  \times [8(i-1) + 5, 8(i-1) + 5]) $
\end{itemize}

\item In addition, we have a fast auxiliary agent $g$ dedicated to deliver the package from last variable gadget to the first clause gadget defined as: 
\begin{itemize}
    \item[$\bullet$] $ g:  (a-1, [-a, 0] \times [8n', 8n'])$,
\end{itemize}
\end{itemize}
Therefore, we have for the size of the grid $n$ that $n\leq 8n'(a\cdot (m'+1)+3)$.
We set the threshold value $t$ to be $n'^3$, i.e. in order to show $n^{1-\varepsilon}$-APX-hardness, we need to show that for any $n^{1-\varepsilon}$-approximation algorithm $A$ it holds, that $\phi$ is satisfiable if and only if the delivery time $t_A(I) \leq n^{1-\varepsilon}(n')^3$, where we define $a:=\left\lceil (n')^{\frac{6}{\varepsilon}} \right\rceil+1$.

\begin{enumerate}
\item[$\implies$:]
$\phi$ is satisfiable, so let $\mathbf{x}$ be a satisfying assignment.\\
Given a variable $x_i$, we refer to the clauses that contain $x_{i, 1}$, $x_{i, 2}$ and $x_{i, 3}$ as $c_{j_1}$, $c_{j_2}$ and  $c_{j_3}$, respectively.

If $\mathbf{x}_i = 0$, we set the positions of both slow agents corresponding to the positive literals $x_{i,1}$ and $x_{i,2}$ to be at the variable gadget of $x_i$, i.e., at $(-2, 8(i-1) + 1)$ and $(-2, 8(i-1) + 3)$, respectively. The position of the slow agent corresponding to the negative literal $x_{i,3}$ is set to be on the vertical line of $c_{j_3}$, i.e., at $(-j_3\cdot a, 8(i-1) + 5)$, which is one unit before the start of $x_{i,3}'$.

If $\mathbf{x}_i = 1$, the position of $x_{i,3}$ is set to be at the variable gadget of $x_i$, i.e., at $(2, 8(i-1) + 5)$. The position of $x_{i,1}$ is set to be on the vertical line of $c_{j_1}$, at $(-j_1\cdot a, 8(i-1) + 1)$, one unit to the right of $x_{i,1}'$ and $x_{i,2}$ is set to be on the line of $c_{j_2}$, so at $(-j_2\cdot a, 8(i-1) + 3)$, one unit to the right of $x_{i,2}'$. 

Thus, for each variable $x_i$, either $x_{i,1}$ and $x_{i,2}$ are positioned at the variable gadget, or $x_{i,3}$ is positioned there. This configuration allows each gadget to be traversed in at most $12$ time units and point $B$ can be reached in time $12n'$. Next, $c_1$ is reached after an additional $\frac{a}{a-1}\leq 2$ time units with auxiliary agent $g$.

Let $x_{i,k}$ be a true literal of $c_1$ in the fulfilling assignment $\mathbf{x}$. We travel with agent $c_1$ to $(-a, 8(i-1) + 2k - 1)$ to the meeting point with $x_{i,k}$ (in at most $8n'$ time steps), and then one unit to the left with $x_{i,k}$ (in one time step) to bridge the gap and reach the movement area of $x_{i,k}'$. The agent $x_{i,k}$ is present at $c_1$, since either $k \in \{1, 2\}$ with $\mathbf{x}_i = 1$, setting $x_{i,k}$ at $c_1$, or $k = 3$ with $\mathbf{x}_i = 0$, also setting $x_{i,k}$ at $c_1$. Then, $x_{i,k}$ is used to reach $c_2$ in $\frac{a-1}{a-1}=1$ time step. This process is repeated for all other clause drones $c_3, \dots, c_m'$, each taking time $2+8n'$. Finally, $y$ is reached by $c_h$, which can be done in $8n'$ time. 
Thus, travelling from
\begin{itemize}
    \item from $s$ to $B$ takes time at most $12n'$
    \item from $B$ to $c_1$ takes time $\leq 2$
    \item from $c_1$ to $c_h$ takes time at most $m' (2+8n')$
    \item from $c_{m'}$ to $y$ takes time at most $8n'$
\end{itemize}
Therefore the total time of the optimal schedule is at most $12n'+2+m'\cdot(2+8n')+8n'$, which is $\leq n'^3$ for sufficiently large instances, therefore $t_A(I) \leq n^{1-\varepsilon}n'^3$.

\item[$\impliedby$:] 
Let $S$ be the schedule returned by algorithm $A$ with duration $t_A(I) \leq n^{1-\varepsilon}\cdot n'^3$. We observe that for sufficiently large instances the grid has size $n=8n'(a\cdot(m'+1)+3) \leq 16n'm'\cdot a \leq n'^3\cdot a$, therefore we can infer
\begin{equation}\label{eq}
    a < n < n'^3\cdot a
\end{equation}
We now show that in $S$, no literal agent can deliver the package over a horizontal distance of $a-1$ or more. 
Suppose there is a literal agent with speed 1 that covers a distance of $a-1$ or more. Then in total, 
$$
\begin{alignedat}{3}
    t_A(I)& > a  && > n'^{\frac{6}{\varepsilon}} \\
    &\iff  a^{\varepsilon} &&> n'^6 \\
    &\stackrel{\mathclap{(\ref{eq})}}{\implies} n^{\varepsilon} &&> n'^6 \\
    &\iff  \frac{1}{n'^6} &&> n^{-\varepsilon} \\
    &\iff  \frac{n}{n'^6} &&> n^{1-\varepsilon} \\
    &\stackrel{\mathclap{(\ref{eq})}}{\implies}  \frac{n'^3 \cdot a}{n'^6} &&> n^{1-\varepsilon} \\
    &\iff  a &&> n^{1-\varepsilon}\cdot n'^3 
\end{alignedat}
$$
therefore $t_A(I) > n^{1-\varepsilon}\cdot n'^3$, which contradicts the bound.

Thus, a single slow literal agent cannot deliver the package directly from a variable gadget to $c_1$ or between clauses without violating our time constraint. Similarly, an agent cannot appear at two points separated by more than $a-1$ units.
Therefore, the package must pass through $B$ and utilize agent $g$, meaning for each variable gadget of $x_i$, either $x_{i,1}$ and $x_{i,2}$ are present for traversing the left path, or $x_{i,3}$ is present for the right path. In the former case, we set $x_i = 0$ (false), and in the latter, $x_i = 1$ (true). 
This construction yields a feasible assignment. It remains to show that this assignment is also satisfying.

To deliver the package to $y$, it must first reach $c_1$, then $c_2$, $c_3$ and subsequently $c_{m'}$ and $c_h$. Since taking the package from $c_{j}$ to $c_{j+1}$ requires a fast agent, we need an agent $x_{i, k}'$, such that $x_{i, k}$ is a literal in clause $c_j$. Then $x_{i, k}$ is needed to bridge the gap from $c_{j}$ to $x_{i,k}'$. 
If $k \in \{1, 2\}$ (indicating that $x_{i, k}$ is a positive literal), then $x_{i,3}$ must be present at the variable gadget. Thus, in our assignment, $x_i = 1$, so $c_{j}$ is satisfied. Conversely, if $k = 3$ (indicating $x_{i, k}$ is a negative literal), then $x_{i,1}$ and $x_{i,2}$ must be at the variable gadget, so $x_i = 0$, satisfying $c_j$. 

This observation translates to all clause gadgets as we derived that they are traversed without using a single literal agent for at least $a-1$ steps. Thus, all clauses are satisfied, making the assignment satisfying. Therefore, $\phi$ is satisfiable.

\end{enumerate}
This rounds up the proof of Theorem \ref{thm:grid_2speed}.\qed

Interestingly, the presented construction also has implications for the setting of initial positioning. 
The idea is to not start immediately at $s$ with the variable gadgets but start with a sufficiently large delay. If the delay is large enough, the literal agents can position themselves regardless of their initial starting point. However, note that this only establishes that the problem is NP-hard; it does not prove that the problem cannot be approximated within $O(n^{1-\varepsilon})$. We lose the inapproximability result due to the inclusion of a significant delay as part of every solution. Additionally, it is worth noting that the results from Theorem \ref{thm:line} do not directly apply to our setting on the unit grid with initial positions. The issue arises because, on a unit grid, for potentially exponentially large input values (with respect to the size of the encoding of the input), we would require $\Omega(P)$ vertices, where 
$P$ is the sum of all elements in the Partition input.

\subsection*{A simple $O(n)$-approximation for DDT-GridR on a unit grid}
We now demonstrate that, in the setting of 
unit grid graphs with different speeds a simple greedy algorithm achieves an \( n \)-approximation. Notably, this result is near-optimal due to Theorem \ref{thm:grid_2speed}, as no $O(n^{1-\varepsilon})$-approximation algorithm for any $\varepsilon > 0$ exists, unless P$=$NP. 
%
%

The proposed algorithm begins by sorting the agents in non-increasing order of speed and iteratively attempts to find an (s-y)-path in the graph induced by the movement areas of only the fastest available agents. 
If a path cannot be found using the fastest agent, the algorithm progressively incorporates slower agents until a valid path is identified. Once a valid path is determined, a corresponding schedule can be constructed using the respective agents. 
The algorithm is presented as follows:

\begin{algorithm}[H]
\caption{Greedy $O(n)$-approximation}
\label{alg:greedy}
\begin{algorithmic}[1]
\STATE Let the agents $a_1, \dots, a_k$ be sorted in non-increasing order of their speeds. 
\FOR{\( i = 1 \) to \( k \)}
    \IF{there exists a path \( P \) from \( s \) to \( y \) in graph \( (V_1 \cup \dots \cup V_i, E_1 \cup \dots \cup E_i) \)}
        \STATE Construct a feasible schedule \( S \) by assigning an agent \( a \in \{a_1, \dots, a_i\} \) to each edge \( e \in P \) such that \( e \in E_a \).
        \STATE Transform \( S \) into schedule \( S^* \) such that each agent is used at most once.   
        \RETURN \( S^* \)
    \ENDIF
\ENDFOR
\RETURN No feasible schedule exists.
\end{algorithmic}
\end{algorithm}
\noindent

During the construction of schedule $S$ in step 4, it is possible for a single agent to be used for multiple disjoint segments of the trip. However, we can convert $S$ into a schedule $S^*$ that uses each agent at most once  without increasing the total delivery time. 
Specifically, for any agent $a$, consider the first and last edge assigned to $a$ in $S$, say $\{u, v\}$ and $\{u', v'\}$, respectively. If there exist edges between $\{u, v\}$ and $\{u', v'\}$ in the schedule $S$ that are assigned to other agents, the schedule can be adjusted to deliver the package from $v$ to $v'$, removing the rest of the schedule in between $\{u, v\}$ and $\{u', v'\}$.  This ensures that each agent is used at most once, and consequently, the package never has to wait for any agent to arrive. Additionally, the agent will deliver the package from $u$ to $v'$ with the minimum total length, as each agent is constrained to a rectangular movement area within the grid. Furthermore, the length of the found ($s$-$y$)-path traversed by $S^*$ is upper-bounded by $n$.

The algorithm runs for at most \( k \) iterations. Each iteration requires \( O(n) \)-time to search for the path $P$ using a depth-first search (DFS). To construct $S^*$, we perform a DFS for each agent $a$ to find a direct path in $(V_a,E_a)$ from the first edge to the last edge assigned to $a$, resulting in $O(n\cdot k)$.   As Step 4-5 are done only once, the total runtime complexity is $O(n \cdot k + k \log k)$, which also includes the time required for sorting the agents.

Suppose the algorithm terminates after $i'$ iterations.
Note that no feasible schedule exists using only the agents \( a_1, \dots, a_{i'-1} \); otherwise, the algorithm would have terminated in an earlier iteration. Consequently, in any feasible schedule, at least one edge (of unit length) must be traversed by an agent with speed at most \( v_{i'} \). Thus, \( OPT \geq 1 / v_{i'} \), with $OPT$ denoting the delivery time of an optimal schedule. 

Since the initial positions of every agent can be chosen and every agent is used at most once, the package never has to wait for any drone to arrive. Therefore, the delivery time of schedule $S^*$ is solely determined by the time spent delivering the package along the corresponding route. The total distance traveled by the package is at most \( n-1 \), and each agent in \( S^* \) has a speed of at least \( v_{i'} \). Therefore, for the total delivery time of schedule $S^*$ we conclude
\[
t(S^*) \leq \frac{n}{v_{i'}}.
\]
Combining this with \( OPT \geq \frac{1}{v_{i'}} \), the approximation guarantee of the algorithm is
\[
\frac{t(S^*)}{OPT} \leq \frac{\frac{n}{v_{i'}}}{\frac{1}{v_{i'}}} = n.
\]

\section{Conclusions}
In this paper we presented new results regarding the complexity of the Drone Delivery Problem (DDT). For DDT on a line, we prove that the problem remains NP-hard when agents are restricted to two different speeds.  For DDT on a grid graph,
We demonstrate that the problem is NP-hard to approximate within 
$O(n^{1-\varepsilon})$ without initial positions for agents with rectangular movement areas and two speeds, while also presenting an $O(n)$-approximation algorithm. 
Prior work by Erlebach et al.~\cite{erlebach:drones} demonstrates that the DDT problem in a general graph is APX-hard. Our results extend this understanding by addressing the case of two speeds on a line and on a unit grid graph, both with and without initial positions.  
Our findings may inspire further exploration into the complexity of collaborative routing problems and potentially completely different fields of research, such as  covering problems.

One limitation of our study is the absence of research into polynomial-time constant approximation algorithms for DDT on a path graph. This gap opens a promising avenue for future research, particularly in both scenarios, with and without initial positions, where the NP-hardness of the latter remains unresolved. It is also of interest to explore subinterval covering problems. 



    \bibliography{refs}
    \bibliographystyle{plainurl}  
\end{document}